\documentclass{elsarticle}
\usepackage{amsfonts}
\usepackage{amssymb}
\usepackage{amsmath,amsthm}
\usepackage{amscd}
\makeatletter
\def\ps@pprintTitle{%
 \let\@oddhead\@empty
 \let\@evenhead\@empty
 \def\@oddfoot{}%
 \let\@evenfoot\@oddfoot}
\makeatother

\newcommand{\beq}{\begin{equation}}
\newcommand{\eeq}{\end{equation}}
\newcommand{\bea}{\begin{eqnarray}}
\newcommand{\eea}{\end{eqnarray}}
\newcommand{\nn}{\nonumber}
\newcommand\noi{\noindent}

\newtheorem{definition}{Definition}

\newtheorem{theorem}{Theorem}

\theoremstyle{definition}

\newtheorem{remark}{\textbf{Remark}}

\begin{document}
\begin{frontmatter}
\title{Beyond the Shannon-Khinchin Formulation: The Composability Axiom and the \\ Universal Group Entropy}
\author[comp,icmat]{Piergiulio Tempesta}
\address[comp]{Departamento de F\'{\i}sica Te\'{o}rica II (M\'{e}todos Matem\'{a}ticos de la f\'isica), Facultad de F\'{\i}sicas, Universidad
Complutense de Madrid, 28040 -- Madrid, Spain \\ and}
\address[icmat]{Instituto de Ciencias Matem\'aticas, C/ Nicol\'as Cabrera, No 13--15, 28049 Madrid, Spain}
\ead{p.tempesta@fis.ucm.es, piergiulio.tempesta@icmat.es}

\begin{abstract}
The notion of entropy is ubiquitous both in natural and social sciences. In the last two decades,  a considerable effort has been devoted to the study of new entropic forms, which generalize the standard Boltzmann-Gibbs (BG) entropy and are widely applicable in thermodynamics, quantum mechanics and information theory. In \cite{Khinchin}, by extending previous ideas of Shannon \cite{Shannon,Shannon2}, Khinchin proposed a characterization of the BG entropy, based on four requirements, nowadays known as the Shannon-Khinchin (SK) axioms.

The purpose of this paper is twofold. First, we show that there exists an \textit{intrinsic group-theoretical structure} behind the notion of entropy. It comes from the requirement of composability of an entropy with respect to the union of two statistically independent systems, that we propose in an axiomatic formulation. Second, we show that there exists a simple universal family of trace-form entropies. This class contains many well known examples of entropies and infinitely many new ones, a priori multi-parametric.  Due to its specific relation with Lazard's universal formal group of algebraic topology, the new general entropy introduced in this work will be called the universal-group entropy. A new example of multi-parametric entropy is explicitly constructed.

\end{abstract}

\date{September 29, 2015}

\end{frontmatter}
\newpage
\tableofcontents

\section{Introduction}
Entropy is a fundamental notion, at the heart of modern science.   In the second half of the twentieth century,  its range of applicability has been extended from the traditional context of classical thermodynamics to new areas such social sciences, economics, biology, quantum information theory, linguistics, etc. More recently, the role of entropy in the theory of complex systems has been actively investigated. From one side, several studies were devoted to axiomatic formulations, aiming at clarifying the foundational aspects of the notion of entropy. From the other side, many researchers pursued the idea of generalizing the classical Boltzmann-Gibbs statistical mechanics. Consequently, a plethora of new entropic forms, designed for extending the applicability of BG entropy to new contexts, was introduced.

The first research line was started by the seminal works by Shannon \cite{Shannon,Shannon2} and Khinchin \cite{Khinchin}. A set of axioms, nowadays called the SK axioms, characterizing uniquely the BG entropy, was introduced (we shall make reference to the formulation of the axioms reported in Appendix A). The   axioms (SK1)-(SK3) represent natural requirements (continuity, maximum principle, independence from zero probability events), that should be satisfied by any functional playing the role of an entropy. Instead, the axiom (SK4) simply characterizes the behaviour of an entropy with respect to the composition of two subsystems, which reduces to additivity in the case of statistical independence of the subsystems.

For long time, additivity was interpreted  as the property that ensures \textit{extensivity}, i.e. the linear dependence of entropy on the number of particles of a system. According to Clausius, extensivity is crucial for an entropy to be thermodynamically admissible.
Surprisingly, the two concepts are completely independent: additivity does not imply, nor is implied by extensivity. In addition, no entropy, irrespectively of being additive or nonadditive, can be extensive in any dynamical regime. For instance, if $W(N)$ is the total number of states of a complex system as a function of the number of its particles $N$, it turns out that a (sufficient) condition for the BG entropy to be extensive is that $W(N)\sim k^{N}$, with $k\in\mathbb{R}_{+}$; however, if $W(N)\sim N^{k}$, it is not.

The second research line, i.e.  the study of generalized entropies and thermostatistics, in which the additivity axiom is explicitly violated, has become an extremely active research area in the last three decades. Since the work of Tsallis \cite{Tsallis1}, many new entropic functionals have been proposed in the literature (see e.g. \cite{Abe2}, \cite{AP}, \cite{Beck}, \cite{BR}, \cite{RC1}, \cite{Hanel,HTGM2}, \cite{Naudts2002}, \cite{Naudts}, \cite{Shafee}, \cite{Tempesta4}, \cite{TsallisCirto}). From the point of view of statistical mechanics, they may generalize the BG entropy in weakly chaotic regimes, when the ergodicity hypothesis is violated and the correlation functions exhibit a non-exponential decay, typically a power-law one. \cite{TEMUCO}. In particular, Tsallis entropy, which is nonadditive, is extensive for special values of the parameter $q$ in regimes where the BG entropy is not \cite{TGS}.

Another source of nonadditive entropies is Information Theory. In this context, generalized entropies can provide more refined versions of Kullback-Leibler-type divergences \cite{KL}, useful for constructing comparative tests of sets of data.

In the study of entanglement, nonadditive entropies arise as necessary alternatives to the von Neumann entropy \cite{Wehrl78}, \cite{RC1}. As usual, a system composed by two subsystems $A$ and $B$ is said to be entangled or separated if its density operator $\rho$ cannot be written as a convex combination of uncorrelated densities, i.e. $\rho\neq \sum_{k} q_{k}\rho_{k}^{A}\otimes\rho_{k}^{B}$. As shown in \cite{Horodecki}, the direct maximization of the von Neumann entropy $S=-\text{Tr} \rho \ln \rho$  does not avoid the detection of fake entanglement even for the case of two spin $\frac{1}{2}$ systems. The need for generalized nonadditive entropies in order to design efficient criteria for separability has been advocated in \cite{RC1}. Recently, the relevance of generalized entropies in processes with quantum memory has been recognized \cite{BCCRR}.

In the present analysis of the mathematical foundations of the concept of entropy, we wish to point out the centrality of the notion of \textit{composability}.   We shall say that an entropy is composable if the following requirements are satisfied. First, given two statistically independent systems $A$ and $B$, the entropy of the composed system $A \cup B$ depends on the entropies of the two systems $S(A)$ and $S(B)$ \textit{only} (apart possibly a set of parameters).  This is the original formulation of the concept, as in \cite{Tbook}.  In addition, we require further properties, as the symmetry of a given entropy in the composition of the systems $A$ and $B$, the stability of total entropy if one of the two systems is in a state of zero entropy, and the associativity of the composition of three independent systems (see axioms (C1)-(C4) below).  With these assumptions, the composability property, in this new, broader sense, is equivalent to the existence of a \textit{group-theoretical structure underlying the notion of entropy}, which guarantees the physical plausibility of the composition process.

All the previous properties ensure that a given entropy can be expressed just in terms of \textit{macroscopic configurations} of a system, without the need for a microscopic description of the associated dynamics. An entropy should be \textit{coarse-grained} \cite{Gell-Mann}; if not so, the concept of entropy would be simply empty. Even the second law of thermodynamics loses any meaning if not referred to the evolution of macroscopic subsystems: the entropy would stay invariant if defined on microscopic configurations.

From the previous discussion, it emerges that keeping only with the first three SK axioms  is a too weak requirement from a thermodynamical perspective: the fourth axiom cannot be simply dropped away.
This situation closely reminds the role of Euclids' fifth postulate of geometry. The replacement (not the mere exclusion) of this postulate with different axiomatic formulations paved the way to non-Euclidean geometries.

Our proposal is \textit{to assume the property of composability in its complete group-theoretical formulation as the generalized form of the fourth axiom}. This assumption is not at all obvious. From this point of view, there are several different perspectives.

a) If we require composability in a strict sense, i.e. for any possible choice of the probability distributions of the given systems, it emerges that the Boltzmann-Gibbs entropy and the Tsallis entropy are the only known cases of trace-form composable entropies (at the best of our knowledge). 

b) At the same time, the notion of composability can be formulated in a weak sense. Indeed, we can impose that the composability axiom be satisfied at least on the \textit{uniform distribution}. This requirement, in thermodynamics, is a fundamental one. Indeed, the uniform distribution emerges when one deals with isolated physical systems at the equilibrium (microcanonical ensemble), or in contact with a thermostat at very high temperature (canonical ensemble). For instance, this is the physical situation occurring when considering high-temperature astrophysical objects.
It turns out that a large set of entropies introduced in the last decades are weakly composable, although not all of them. Again, our formulation unravels the intrinsic group-theoretical content of the concept of entropy: in all cases, the composability requirement amounts to the existence of a \textit{group law} for the composition process, defined either over the whole set of probability distributions (strict composability) or just over the uniform one (weak composability).


In this work, we present two main results. First, we encode composability into a general axiomatic formulation, required for an entropy to be considered admissible.


Second, we introduce a very general trace-form entropic form. We shall call it the \textit{universal-group entropy}, due to its relation with Lazard's universal formal group. This entropy offers infinitely new cases of trace-form weakly composable entropies, therefore satisfying the first and third SK axioms. Also, all known entropies are directly related with our new entropic form.

The universal group entropy depends a priori on a (possibly) infinite number of independent parameters.
Special attention will be paid to the recently introduced $S_{c,d}$ entropy \cite{Hanel} and to group entropies \cite{Tempesta4}, respectively. We will show that both cases are related to the universal-group theoretical framework in a very direct way.
Besides, multi-parametric entropies can be introduced, not related to any of the known entropies. As an interesting example, the three-parametric $S_{\alpha,\beta, q}$ entropy is proposed.

The paper is organized as follows. In Section 2, the group structure behind the notion of entropy is introduced. In Section 3, the theory of formal groups, and in particular that of the universal formal group, is sketched. In Section 4, the universal-group entropy is defined and its main properties are studied in Section 5. In Section 6, a new example of three-parametric entropy is proposed. In Section 7, the Legendre-type structure and some thermodynamical aspects of the theory are discussed.  In Section 8, the concept of thermodynamic limit of an entropy is discussed. In Section 9, some open problems and possible related research lines are briefly proposed. In the Appendix A, the formulation of the SK axioms adopted in the paper is reported.

We expect that the present general formulation of the notion of entropy, apart its intrinsic theoretical interest, can have further use both in classical and quantum information theory and beyond.

\section{The group-theoretical content of the notion of entropy}


The purpose of this section is to establish the existence of a group-theoretical structure underlying the concept of (generalized) entropy. To this aim, we propose first a new formulation of the notion of composability. We shall assume that the entropies we consider are sufficiently regular functions  $S(p_1,\ldots,p_W)$  defined for any integer $W\geq 1$ in the set $\mathcal{P}_W$ of all discrete probability distributions with $W$ entries, and taking values in
 $\mathbb{R}^{+}\cup \{0\}.$

\subsection{Strong and weak composability}
\begin{definition}\label{composab} An entropy $S$ is strongly (or strictly) \textit{composable} if there exists a smooth function of two real variables $\Phi(x,y)$ such that

\text{(C1)}
 \beq
S(A \cup B)=\Phi(S(A),S(B);\{\eta\}) \label{C1}
\eeq
where $A\subset X$ and $B\subset X$ are two statistically independent systems, each defined in terms of an arbitrary probability distribution $\{p_{i}\}_{i=1}^{W}$, and $\{\eta\}$ is a possible set of real continuous parameters. In addition, $\Phi(x,y)$ satisfies the following properties:

(C2) Symmetry:
\beq
\Phi(x,y)=\Phi(y,x); \label{C2}
\eeq

(C3) Associativity:
\beq
\Phi(x,\Phi(y,z))=\Phi(\Phi(x,y),z); \label{C4}
\eeq

(C4) Null-composability:
\beq
\Phi(x,0)=x. \label{C3}
\eeq

\end{definition}
The symmetry property is an obvious requirement. The null-composability ensures that the composition of a system with another system in a state of zero entropy cannot affect thermodynamics.  The associativity property is a new, essential point: it guarantees the composability of more than two systems.

\begin{remark}
Notice that, assuming $\Phi(x,y)=x+y+\sum_{kl}c_{kl} x^k y^l$, the existence of a power series $\varphi(x)$ such that $\Phi(x,\varphi(x))=0$, i.e. playing the role of an inverse, is a direct consequence of the previous axioms (this observation is also valid when working with formal power series only). The previous requirements amount to say that, from a mathematical point of view, $\Phi(x,y)$ defines a formal \textit{group law} over the reals. For instance, in the case of the Boltzmann entropy it is nothing but the additive group over $\mathbb{R}$.

However, when we restrict to values $x,y\in \mathbb{R}^{+}\cup \{0\}$, as is the case when $\Phi(x,y)$ is computed over standard entropies, we do not have an inverse.
\end{remark}
\begin{remark}
Condition \eqref{C1} can be formulated in more general terms. Consider the case of a composite system $A\cup B$ arising from two systems not statistically independent, with a conditional probability distribution $p_{ij}(B\mid A):= p_{ij}(A \cup B)/p_{i}(A)$. Here $p_{ij}(A \cup B)$, $i=1,\ldots,W_{A}$, $j=1,\ldots, W_{B}$ denotes the joint probability distribution for the composite system $A \cup B$, and $p_{i}(A)$ is the marginal probability distribution $p_{i}(A)=\sum_{j=1}^{W_{B}}p_{ij}(A,B)$. In this context we postulate the relation
\beq
S(A \cup B)=\Phi(S(A),S(B\mid A);\{\eta\}), \label{C1bis}
\eeq
where $S(B\mid A)$ denotes the conditional entropy associated with the conditional distribution $p_{ij}(B\mid A)$. Equation \eqref{C1bis} reduces to the relation \eqref{C1} in the case of statistically independent systems.  The relation \eqref{C1bis} generalizes the axiom (SK4).
\end{remark}
We can also propose a \textit{weak formulation} of composability.
\begin{definition} We shall say that an entropy is weakly composable if the properties (C1)--(C3) of Definition \ref{composab} are satisfied at least when the probability distributions of the two statistically independent systems $A$ and $B$ are both uniform, and property ({C4}) holds in general.
\end{definition}
This weak formulation is in fact satisfied by infinitely many generalized entropies, as we shall prove. For instance, the function $\Phi(x,y)$, and consequently its group structure, can be typically constructed starting from generalized logarithms. This second point of view has been first proposed in \cite{Tempesta4} for a class of logarithms coming from difference operators. Again, the weak composability requirement implies the existence of a group law.
\subsection{On the relation between the SK axioms, composability and admissible trace-form entropies}
The SK axioms and the composability axiom are strictly related. First observe that, for both formulations of composability, the continuity of a given entropy $S$ with respect to its arguments, as required by the axiom (SK1), implies the continuity of $\Phi(x,y)$ at least for $x,y\in \mathbb{R}^{+}\cup\{0\}$.

In the strong formulation of the notion of composability, the property \eqref{C1bis}, valid for the Boltzmann and Tsallis entropies, is equivalent to the  axiom (SK4). However, in both formulations, we still have to impose the further requirement of \textit{strict concavity} (to ensure that axiom (SK2) is fulfilled).




Motivated by the previous discussion, we analyze the requisites for an entropy to be considered \textit{admissible}, i.e. relevant both from a physical and information-theoretical point of view. In the recent literature, usually admissible entropies are considered to be those satisfying just the axioms (SK1)--(SK3). Our point of view is therefore more demanding: \textit{A necessary condition for an entropy $S$ to be \textit{admissible} is that it satisfies the axioms (SK1)--(SK3) and is (at least) weakly composable.}

Indeed, there exist entropies which do satisfy the axioms (SK1)--(SK3), but are not even weakly composable.
Needless to say, the strong composability property is much more suitable for thermodynamical purposes than its weak formulation.

The general problem of classifying the admissible entropies will be discussed in the forthcoming sections.
\section{Formal group laws}

The theory of formal groups \cite{Boch}  offers the natural language for formulating our approach to the theory of generalized entropies. Formal groups have been intensively investigated in the last decades, especially for their prominent role in fields as algebraic topology, the theory of elliptic curves, and arithmetic number theory \cite{Haze}, 
\cite{Tempesta3}. Here we briefly recall only some salient aspects, necessary in the subsequent discussion.

Let $R$ be a commutative ring  with identity, and $R\left\{ x_{1},\text{ }%
x_{2},..\right\} $ be the ring  of formal power series in the variables $x_{1}$, $x_{2}$,
... with coefficients in $R$. A \textit{commutative one-dimensional formal group law}
over $R$ \cite{Boch} is a formal power series in two variables $\Psi \left( x,y\right)  \in R\left\{
x,y\right\} $ of the form $\Psi \left( x,y\right)=x+y+ \text{terms of higher degree}$, such that
\begin{equation*}
i)\qquad \Psi \left( x,0\right) =\Psi \left( 0,x\right) =x
\end{equation*}%
\begin{equation*}
ii)\qquad \Psi \left( \Psi \left( x,y\right) ,z\right) =\Psi \left( x,\Psi
\left( y,z\right) \right) \text{.}
\end{equation*}
When $\Psi \left( x,y\right) =\Psi \left( y,x\right) $, the formal group law is
said to be commutative (the existence of an inverse formal series $\varphi
\left( x\right) $ $\in R\left\{ x\right\} $ such that $\Psi \left( x,\varphi
\left( x\right) \right) =0$ follows from the previous definition).

The simplest examples are the additive formal group law $\Psi(x,t)=x+y$ and the multiplicative one $\Psi(x,y)=x+y+xy$.

The previous definition can be naturally extended to the case of $n$-dimensional formal group laws.

The relevance of formal groups relies first of all on their close connection with group theory. Precisely, a formal group law $\Psi(x,y)$ defines a functor $F: \bf{Alg}$$_{R} \longrightarrow \bf{Group}$, where $\textbf{Alg}$$_R$ denotes the category of commutative unitary algebras over $R$ and $\textbf{Group}$ denotes the category of groups \cite{Haze}. The functor $F$ is by definition the formal group (sometimes called the formal group scheme) associated to the formal group law $\Psi$.

%

%
%
%
%
%

As is well known, over a field of characteristic zero, there exists an equivalence of categories between Lie algebras and formal groups \cite{Serre}.

%
%
%

Any $n$--dimensional formal group law defines a $n$--dimensional Lie algebra over the same ring $R$ by means of the identification
\begin{equation}
[x,y]=\Psi_2(x,y)-\Psi_2(y,x),
\end{equation}
where $\Psi_2(x,y)$ denotes the quadratic part of the formal group law $\Psi(x,y)$.  This equivalence of categories is no longer true in a field of characteristic $p\neq0$.



\subsection{The Lazard  formal group} The main algebraic structure we need is provided by the following construction.
Let  $B=\mathbb{Z}[b_{1},b_{2},...]$ be the  ring of integral polynomials in infinitely many variables. We shall consider the series (formal group logarithm)
\beq
F\left( s\right) = \sum_{i=0}^{\infty} b_i \frac{s^{i+1}}{i+1}
\label{I.1},
\eeq
with $b_0=1$. Let $G\left( t\right)$ be its compositional inverse (the formal group exponential):
\beq
G\left( t\right) =\sum_{k=0}^{\infty} a_k \frac{t^{k+1}}{k+1} \label{I.2}
\eeq
so that $F\left( G\left( t\right) \right) =t$. We have $a_{0}=1, a_{1}=-b_1, a_2= \frac{3}{2} b_1^2 -b_2,\ldots$.
The Lazard formal group law \cite{Haze} is defined by the formal power series
\[
\Phi \left( s_{1},s_{2}\right) =G\left( F\left(
s_{1}\right) +F\left( s_{2}\right) \right).
\]

It has a great relevance in many branches of mathematics, as algebraic topology, cobordism theory, etc. \cite{Haze}, \cite{BK},  \cite{BMN}, \cite{Quillen}.


The coefficients of the power series $G\left( F\left( s_{1}\right) +F\left(
s_{2}\right) \right)$ lie in the ring $B \otimes \mathbb{Q}$ and generate over $\mathbb{Z}$ a subring $A \subset B \otimes \mathbb{Q}$, called the Lazard ring $L$.

The following important results, due to Lazard, hold. First, for any commutative one-dimensional formal group law over any ring $R$, there exists a unique homomorphism $L\to R$ under which the Lazard group law is mapped into the given group law (the so called universal property of the Lazard group).

At the same time, for any commutative one-dimensional formal group law $\Psi(x,y)$ over any ring $R$, there exists a series $\psi(x)\in R[[x]] \otimes \mathbb{Q}$ such that
\[
\psi(x)= x+ O(x^2), \quad \text{and} \quad \Psi(x,y)= \psi^{-1}\left(\psi(x)+\psi(y)\right)\in R[[x,y]]\otimes \mathbb{Q}.
\]


In \cite{BMN} and \cite{Nov}, it has been pointed out that the Lazard universal group law is fundamental in the discussion of both the classical and modern theory of unitary cobordisms. A beautiful connection with combinatorial Hopf algebras and Rota's umbral calculus  was established in \cite{BCRS}. A combinatorial approach has been proposed in \cite{AB}.

In the papers \cite{Tempesta1}, \cite{Tempesta3} a connection between the Lazard universal formal group and the theory of Dirichlet series has been established. In this context, the \textit{Universal Bernoulli polynomials} were also introduced and their remarkable properties studied in  \cite{Tempesta5} (see also \cite{Tempesta2}). In \cite{MT},  the universal Bernoulli polynomials have been related to the theory of hyperfunctions of one variable by means of an extension of the classical Lipschitz summation formula to negative powers.

\section{The universal-group entropy.}
\noi  We are now able to construct a very general family of \textit{trace-form entropies} with relevant thermodynamical properties.
\subsection{The universal-group entropy}
\begin{definition} \label{def1}
Let $\{p_i\}_{i=1,\cdots,W}$, $W\geq 1$, with $\sum_{i=1}^{W}p_i=1$, be a discrete probability distribution. Let
\begin{equation}
G\left( t\right) =\sum_{k=0}^{\infty} a_k \frac{t^{k+1}}{k+1} \label{I.2}
\end{equation}
be a real analytic function, where $\{a_{k}\}_{k\in\mathbb{N}}$ is a sequence of parameters, with $a_{0}\neq 0$, such that the function  $S_{U}:\mathcal{P}_{W}\rightarrow \mathbb{R}^{+}\cup\{0\}$, 
defined by
\beq
S_{U}(p_1,\ldots,p_W):=k_{\mathcal{B}} \sum_{i=1}^{W}p_i \hspace{1mm} G\left(\ln \frac{1}{p_i}\right), \label{entropy}
\eeq
is a concave one. This function will be called the \textit{universal-group entropy}.
\end{definition}
\noi The name of the entropy \eqref{entropy} is obviously reminiscent of its direct connection with Lazard's construction of a universal formal group law. Indeed, $G(t)$ is a group exponential.
The function $G(t)$, or equivalently the real sequence $\{a_{k}\}_{k\in\mathbb{N}}$, encodes all the main features of the entropy \eqref{entropy}. A simple way to construct an important subclass of concave entropies is provided by the following observation.
\begin{remark}
\noi The condition
\beq
 a_{k} >    (k+1)  a_{k+1} \qquad \forall k\in\mathbb{N} \qquad \text{with } \{a_{k}\}_{k\in\mathbb{N}}\geq 0 \label{conc}
\eeq
is sufficient to ensure that the series $G(t)$ is absolutely and uniformly convergent with a radius $r=\infty$ and that $S_{U}[p]$ is a strictly concave functional (see the proof of Theorem \ref{theoradm}). Although certainly restrictive, condition \eqref{conc} is satisfied by many of the entropies known in the literature.
\end{remark}
\begin{remark} The universal-group entropy depends on the infinite set of parameters $a_k$, which are a priori independent, apart for the existence of specific constraints, as in eq. \eqref{conc}. To recover known cases of one-parametric or two-parametric entropies, depending let's say on the parameters $q_1$ and $q_2$, we shall have that $a_k=f_{k}(q_1,q_2)$. Also, notice that the (apparently) more general case of an entropy of the form
\beq
S_{U}=k_{\mathcal{B}} \sum_{i=1}^{W}p_i \hspace{1mm} G\left(\ln \frac{1}{{p_i}^{c}}\right), \label{centropy}
\eeq
where $c$ is a positive constant, can be easily recast in the language of Definition \ref{def1}. Indeed, it is sufficient to consider a formal group exponential $\widetilde{G}(t)=\sum_{k=0}^{\infty} \alpha_k \frac{t^{k+1}}{k+1}$, with  $\alpha_{k}=a_{k} c^{k}$ or, equivalently, to require a modified condition \eqref{conc} of the form $a_{k}> c (k+1) a_{k+1}$, with $\{a_{k}\}_{k\in\mathbb{N}}\geq 0$.

\end{remark}
\begin{remark}
The entropy \eqref{entropy} is \textit{trace-form}. Although very general, the family of trace-form entropies does not include other functional forms, like Renyi's entropy \cite{Renyi}, very useful in several applications. The fact that the entropy \eqref{UGE} is trace-form, with a group exponential at least piecewise differentiable,  ensures \textit{Lesche stability} \cite{Lesche}, which is a desirable property of any entropic functional, and more generally of any physical observable. Essentially, Lesche stability is the property of uniform continuity of an entropic functional in the space of probability distributions. It guarantees that a small variation of the set of probabilities produces a small change of entropy: given two distributions $\{p_i\}_{i=1,\ldots,W}$ and $\{p'_i\}_{i=1,\ldots,W}$ whose values are slightly different,
\beq
\forall \epsilon \quad \exists \delta >0 \quad s. t. \quad \|p-p'\|_1 < \delta \Longrightarrow \left| \frac{S[p]-S[p']}{S_{max}}\right|< \epsilon,
\eeq
where, given a vector $x$, $\|x\|_1$ denotes the $l_1$ norm and $S_{max}$ denotes the maximum value of the entropy.

This requirement, crucial for physical applications, is not so relevant in other contexts as, for instance, Information Theory. 
\end{remark}

\section{Main properties of $S_{U}$ entropy}
\noi In the following, we shall focus on the most relevant properties of the entropy \eqref{entropy}, directly coming from its group-theoretical formulation, and on its relationship with other entropic functionals.
\subsection{Composability}

\begin{theorem}
The entropy $S_{U}[p]$ is weakly composable.
\end{theorem}
\begin{proof}
For sake of clarity, we shall formulate the proof in a general setting. Let $\{p_{i}^{A}\}_{i=1}^{W_{A}}$ and $\{p_{j}^{B}\}_{j=1}^{W_{B}}$ two sets of probabilities associated with two statistically independent systems $A$ and $B$. The joint probability is given by
\[
p_{ij}^{A \cup B}=p_{i}^{A}\cdot p_{j}^{B}
\]
The total number of states of the composed system (including states with possibly null probability),  is $W_{AB}=W_{A}W_{B}$. We have

\bea \label{seq}
\nonumber  S_{U}(A \cup B):&=&k_{\mathcal{B}} \sum_{i=1}^{W_{A}}\sum_{j=1}^{W_{B}}p_{ij}^{A \cup B} \hspace{1mm} G\left(\ln\frac{1}{p_{ij}^{A \cup B}}\right)= \\
\nn &=& k_{\mathcal{B}}\sum_{i=1}^{W_{A}}\sum_{j=1}^{W_{B}}p_{i}^{A}\cdot p_{j}^{B} \hspace{1mm} G\left(\ln\frac{1}{p_{i}^{A}}+\ln\frac{1}{p_{j}^{B}}\right)= \\
\nn &=& k_{\mathcal{B}}\sum_{i=1}^{W_{A}}\sum_{j=1}^{W_{B}}p_{i}^{A}\cdot p_{j}^{B} \hspace{1mm} G\left(t_1+t_2\right)= \\
 \nn &=& k_{\mathcal{B}}\sum_{i=1}^{W_{A}}\sum_{j=1}^{W_{B}}p_{i}^{A}\cdot p_{j}^{B} \hspace{1mm} G\left(F(s_1)+F(s_2)\right)=  \\ \nn
 \nn &=& k_{\mathcal{B}}\sum_{i=1}^{W_{A}}\sum_{j=1}^{W_{B}}p_{i}^{A}\cdot p_{j}^{B} \hspace{1mm}\Phi\left(s_1,s_2\right)=\\
 \nn &=& k_{\mathcal{B}}\sum_{i=1}^{W_{A}}\sum_{j=1}^{W_{B}}p_{i}^{A}\cdot p_{j}^{B} \hspace{1mm}\left(s_1+s_2+ \sum_{k,m=1}^{\infty} c_{k m} s_{1}^{k} s_{2}^{m}\right)= \\ \nn
  \nn &=& k_{\mathcal{B}}\sum_{i=1}^{W_{A}}\sum_{j=1}^{W_{B}}p_{i}^{A}\cdot p_{j}^{B} \hspace{1mm} \left[G\left(\ln\frac{1}{p_{i}^{A}}\right)+G\left(\ln\frac{1}{p_{j}^{B}}\right)+\ldots \right]\nn
  .\\
\eea
These equalities hold in full generality. Now, consider the Boltzmann and Tsallis composition laws, corresponding to the additive case (i.e. $c_{km}=0$ $\forall k,m$) and to the case $c_{11}\neq 0$ and $c_{km}=0$ for $(k,m)\neq(1,1)$, respectively. The last expression in formula \eqref{seq} yields immediately
\beq
 S_{U}(A+B)= \Phi\left(S_{U}(A),S_{U}(B)\right) \label{finale}.
\eeq
Observe that whenever $c_{km}\neq 0$ for $(k, m)\neq(1,1)$, a priori there is no guarantee that formula \eqref{finale} holds in general, i.e. for any possible choice of the probability distributions of the two subsystems. 

However, in the case that $\{p_{i}^{A}\}$ and $\{p_{j}^{B}\}$ are both the \textit{uniform distribution}, formula \eqref{finale} holds for the whole family of entropies represented by \eqref{entropy}.

  The  formal power series $\phi(x,y)= G(F(x)+F(y))$ for any choice of $G$ defines a formal group law. It verifies automatically the conditions of symmetry, null composability and transitivity. The thesis follows.
\end{proof}
In the case of an entropy of the form \eqref{centropy}, the same analysis holds.

\begin{theorem}
The universal group entropy satisfies the first three SK axioms. \label{theoradm}
\end{theorem}


(SK1). The group exponential $G(t)$ is a real analytic function of $t$. Consequently, the universal group entropy is (at least) a continuous function of its arguments $(p_1,\ldots,p_w)$.

(SK2). The entropy $S_{U}[p]$ is supposed to be concave by definition. Nevertheless, we wish to prove here that condition \eqref{conc}, which allows to construct a large subclass of entropies of the form \eqref{entropy},  is \textit{sufficient} to guarantee concavity. To this aim, consider the quantity $\sum_{k=0}^{\infty}\frac{a_{k}}{k+1} \hspace{1mm} x \left(\ln \frac{1}{x}\right)^{k+1}$. By imposing strict concavity we get
\beq
-\frac{1}{x}\{a_{0}-a_{1}+(a_{1}-2a_{2})\ln\frac{1}{x}+(a_{2}-3a_{3})\left(\ln\frac{1}{x}\right)^{2}+ (a_{3}-4a_{4})\left(\ln\frac{1}{x}\right)^3+\ldots\} < 0. \label{ineq}
\eeq
Condition $a_{k}  > (k+1) a_{k+1}$ ensures that the inequality \eqref{ineq} is satisfied. Consequently, the associated entropy \eqref{entropy} is strictly concave in its space of parameters. Other choices of the sequence $\{a_{k}\}_{k\in\mathbb{N}}$ are clearly possible.

(SK3). Since by construction $G(0)=0$, and $\lim_{x\rightarrow 0} x \left(\ln \frac{1}{x}\right)^k=0$, it follows that $S_{U}(0)=0$. Similarly, $S_{U}(1)=0$.


\subsection{Extensivity}

We shall study here the extensivity properties of the universal-group entropy.
Consider the uniform probability distribution (i.e. $p_{i}=1/W$ for all $i=1,\ldots,W$). We have that
\beq
S_{U}[W]=k_{\mathcal{B}}G\left(\ln W\right) \sim N \Longleftrightarrow W(N)\sim\exp \left(F(N)\right). \label{occlaw}
\eeq
The expression for $W(N)$ so determined can be computed explicitly. Indeed, as a formal series $F(s)$ is the compositional inverse of $G(t)$: it can be constructed by means of the Lagrange inversion principle.

From a physical point of view, one should also ensure that  $W(N)$ be interpretable as an occupation law. A sufficient condition is that $W(N)$, as a real function, be defined for all $N\in\mathbb{N}$, with $\lim_{N\to\infty} W(N)=\infty$. These requirements usually restrict the space of allowed parameters. As an example, for the case of Tsallis entropy we have that the corresponding $W(N)$ is a well defined occupation law for the values of the entropic parameter $q<1$.

Therefore, provided the previous condition is satisfied, when choosing an occupation law  $W(N)$ of the form \eqref{occlaw}, there exists a representation of the universal-group entropy which is extensive for all systems whose phase space grows according with the selected occupation law.
This notable fact ensures the large applicability of the entropy \eqref{entropy} in thermodynamical contexts.



\subsection{Relation with other entropies}


We shall discuss now how the previous formalism applies to the classification of known entropies.

Let us first observe that the universal entropy admits the following formal  decomposition
\beq
S_{U}[p]= \sum_{k=1}^{\infty}a_{k}'S_{k}[p] \label{UGE}
\eeq
with $a_{k}'=a_{k-1}/k$, in terms of a  set of elementary functionals
$S_{k}[p]:=k_{\mathcal{B}} \sum_{i=1}^{W} p_i  \left( \ln \frac{1}{p_i}\right)^{k}$.

\noi a) The \textit{Boltzmann-Gibbs entropy}
\beq
S_{\mathcal{B}}[p]=k_{\mathcal{B}}\sum_{i=1}^{W}p_i \ln \frac{1}{p_i} \label{BG}
\eeq
is obtained by means of the choice $G(t)=t$, $c=1$ (i.e. for $a_{1}=1,a_{i}=0$, $\forall i=2,3,\ldots$).

\noi b) The \textit{Tsallis entropy} \cite{Tsallis1}
\beq
S_{q}[p]=k_{\mathcal{B}}\frac{\sum_{i=1}^{W}  p_{i}^{q}-1}{1-q} \label{TS}
\eeq
for $q <1$ corresponds to the formal group exponential $G(t)=\frac{\exp[(1-q)t]-1}{1-q}$, $c=1$. It can be decomposed as
\beq
S_{q}[p]= k_{\mathcal{B}} \sum_{i=1}^{W} p_i \{ \ln \frac{1}{p_i}+\frac{1}{2}(1-q)  \left(\ln \frac{1}{p_i}\right)^2+  \frac{1}{6}(1-q)^2 \left(\ln \frac{1}{p_i}\right)^3+\ldots \}.
\eeq
When $q > 1$, $G(t)=\frac{\exp[(q-1)t]-1}{q-1}$ and a very similar expansion holds:
\beq
S_{q}[p]= k_{\mathcal{B}} \sum_{i=1}^{W} p_i \{ \ln \frac{1}{p_i}+\frac{1}{2}(q-1)  \left(\ln \frac{1}{p_i}\right)^2+\frac{1}{6}(q-1)^2 \left(\ln \frac{1}{p_i}\right)^3+\ldots \}.
\eeq
\noi c) The \textit{Kaniadakis entropy} \cite{Kaniad1}
\beq
S_{\kappa}[p]=k_{\mathcal{B}}\sum_{i=1}^{W} p_i \frac{p_{i}^{-\kappa}-p_{i}^{\kappa}}{2 \kappa} \label{Kaniad}
\eeq
is obtained by means of the choice $G(t)=\frac{\exp(\kappa t)-\exp(-\kappa t)}{2 \kappa}$, $c=1$. The decomposition \eqref{UGE} of the Kaniadakis entropy is given by
\beq
S_{\kappa}[p]=k_{\mathcal{B}} \sum_{i=1}^{W} p_i \left\{ \ln \frac{1}{p_i}+\frac{1}{3!}\kappa^2  \left(\ln \frac{1}{p_i}\right)^3+\frac{1}{5!}\kappa ^4  \left(\ln \frac{1}{p_i}\right)^5+ \frac{1}{7!}\kappa^6 \left(\ln \frac{1}{p_i}\right)^7+ \ldots \right\}
\eeq
where $-1<k\leq1$.


\noi d) Another important case is the \textit{$S_{c,d}$ entropy}, introduced in \cite{Hanel}. This beautiful entropic form was obtained by taking into account two scaling laws that emerge from the requirement of the first three SK axioms.  It reads
\beq
S_{c,d} =\frac{e}{1-c+c d} \sum_{i=1}^{W} \Gamma(1+d,1-c \ln p_{i})-\frac{c}{1-c+c d}. \label{SHanel}
\eeq
Here $\Gamma(s,x)$ denotes the upper incomplete Gamma function (and $k_{\mathcal{B}}=1$), $c\in(0,1]$, $d\in\mathbb{R}$.

The authors analyzed essentially all entropies known in the literature that satisfy the axioms (SK1)--(SK3), with the exception of group entropies, and observed that these entropies can be considered as particular cases of the $S_{c,d}$ entropy, for a suitable choice of $(c,d)$ (see also \cite{HTG2014} for a recent discussion of the role of the $S_{c,d}$ entropy). These parameters appear as the exponents characterizing the two scaling laws.

Consequently, in this Section, we will not compare exhaustively all the other cases already considered in \cite{Hanel} (as e.g. the Anteneodo-Plastino \cite{AP} and Shafee entropy \cite{Shafee}). Instead, we shall discuss directly the relationship between the $S_{c,d}$ entropy and the universal-group entropy.

\noi Consider the identity
\[
\Gamma(s,x)=\Gamma(s)-\gamma(s,x)
\]
where $\gamma(s,x)$ is the lower incomplete gamma function, as well as the series expansion of $\gamma(s,x)$ for real positive values of its arguments. We obtain that
\beq
\Gamma(1+d, 1-c \ln p_{i})= \Gamma(1+d)+t^{d}\sum_{k=0}^{\infty} \delta_{k} \frac{t^{k+1}}{k+1}, \label{gamma}
\eeq
where $\Gamma(a)$ is the Euler Gamma function,
\beq
\delta_{k}=\frac{(-1)^{k+1}(k+1)}{k!(k+d+1)}, \qquad t= \left[\ln \frac{ e}{{p_{i}}^{c}}\right]. \label{ti}
\eeq
In the subsequent considerations, we shall restrict to the case $d\in\mathbb{N}$.
\noi To perform our analysis we shall use the identity, valid for $d\in\mathbb{N}$:

\beq
\int_{K}^{\infty}t^{d} e^{-t}dt=e^{-K} \sum_{n=0}^{d}\frac{\prod_{j=0}^{n}(d-j+1)}{d+1}K^{d-n}. \label{identity}
\eeq
This identity (new, to the best of our knowledge) can be proven by a direct computation.

The identity \eqref{identity} allows to expand the entropy in terms of the set $\{S_{k}[p]\}$.
\noi In general, we have
 \bea
\nn S_{c,d}&=&\frac{1}{1-c+c d}\sum_{i=1}^{W} {p_i}^{c} \left\{\sum_{k=0}^{d}\frac{1}{d+1}\prod_{j=0}^{k}(d-j+1)\sum_{n=0}^{d-k}\binom{d-k}{n}\left(\ln \frac{1}{{p_i}^{c}} \right)^{n} \right\} \\&-&\frac{c}{1-c+c d}, \qquad d\in\mathbb{N}. \label{genexp}
\eea
Let us see some particular cases of the previous formula. We easily recover the two examples of \cite{Hanel}, i.e.
\[
S_{1,1}[p]=1+\sum_{i}p_i \ln \frac{1}{p_i}, \qquad  S_{1,2}[p]=2\left(1+ \sum_{i} p_{i}\ln \frac{1}{p_i}\right)+\frac{1}{2} \sum_i p_i \left(\ln \frac{1}{p_i}\right)^2.
\]
\noi For $c$ arbitrary, let us write explicitly, for instance, the functionals corresponding to $d=3$:
\beq
S_{c,3}[p]=\frac{1}{1+2 c}\sum_{i=1}^{W} {p_i}^{c} \left\{16+15 \ln\frac{1}{{p_i}^{c}}+6 \left(\ln\frac{1}{{p_i}^{c}}\right)^2+\left(\ln\frac{1}{{p_i}^{c}}\right)^3\right\}-\frac{c}{1+2 c },
\eeq
and $d=5$:
\bea
\nn S_{c,5}[p]&=&\frac{1}{1+4c}\sum_{i=1}^{W} {p_i}^{c} \bigg\{326 +325 \ln\frac{1}{{p_i}^{c}}+160 \left(\ln\frac{1}{{p_i}^{c}}\right)^2 +50 \left(\ln\frac{1}{{p_i}^{c}}\right)^3 \\
 &+&10 \left(\ln\frac{1}{{p_i}^{c}}\right)^4 +\left(\ln\frac{1}{{p_{i}}^{c}}\right)^5\bigg\}-\frac{c}{1+4c}.
\eea

From the previous analysis, it emerges that the $S(c,d)$ entropy fits into the class \eqref{entropy} in the two cases $(c=1,d\in\mathbb{N})$ (we can get rid of the constant term in the expansion) and  $(c>0, d=0)$.

e) The $S_{\delta}$ entropy has been introduced in \cite{Tbook} and independently in \cite{Ubriaco}, and recently discussed in \cite{TsallisCirto}  in relation with black-holes thermodynamics:
\beq
S_{\delta}=k_{\mathcal{B}}\sum_{i=1}^{W}p_{i}\left(\ln \frac{1}{p_{i}}\right)^{\delta}, \qquad 0<\delta\leq (1+\ln W).
\eeq


The underlying algebraic structure can be analyzed on the uniform distribution (weak composability). We get easily the associated function $\Phi(x,y)=\left[x^{1/\delta}+y^{1/\delta}\right]^{\delta}$. Notice that the function $\Phi(x,y)$ has not an expansion in terms of a formal power series around $(x,y)=(0,0)$, unless $\delta=1$, which is the only strictly composable case. Also, $\Phi(x,y)$ does not define a group law over the reals, but simply a monoid, except for $\delta\in\mathbb{N}$, $\delta$ odd. A similar analysis can be performed for the case of the entropic functional
\beq
S_{q,\delta}=k_{\mathcal{B}}\sum_{i=1}^{W}p_{i}\left(\ln_q \frac{1}{p_{i}}\right)^{\delta}.
\eeq
It reduces to \eqref{TS} for $q\in \mathbb{R}$ and $\delta=1$.


f) The \textit{Borges-Roditi entropy} is a two-parametric entropy, introduced in \cite{BR}. The associated generalized logarithm reads
\beq
Log_{a,b}(x)=\frac{x^{a}-x^{b}}{a-b}, \label{Bor}
\eeq
which reproduces Abe's entropy \cite{Abe2} for $a=\sigma-1, b=\sigma^{-1}-1$. It also generalizes the entropy \eqref{Kaniad}.
The entropy \eqref{Bor}  for a suitable choice of $a,b$ satisfies the first three SK axioms (however, it is not related to the $S_{c,d}$ entropy). The related group exponential is
\beq
G_{\mathcal{A}}(t)=\frac{e^{at}-e^{bt}}{a-b}\label{Abel}.
\eeq
The formal group corresponding to (\ref{Abel}), giving the interaction rule for this class, is known in the literature as the \textit{Abel formal group}, defined by \cite{BK}
\beq
\Phi_{\mathcal{A}}(x,y)=x+y+\beta_1 xy+ \sum_{j>i} \beta_i\left(xy^{i}-x^{i}y\right). \label{AFG}
\eeq
The coefficients $\beta_n$ in (\ref{AFG}) can be expressed as polynomials in $a$ and $b$ (see Proposition 3.1 of \cite{BK}):
\[
\beta_{1}=a+b,\qquad \beta_n=\frac{(-1)^{n-1}}{n!(n-1)}\prod_{\overset{i+j=n-1}{i,j\geq 0}} (ia+jb), \qquad n>1.
\]
The Borges-Roditi entropy possesses the following expansion:
\beq
S_{BR}[p]=k_{\mathcal{B}} \sum_{i=1}^{W} p_i \left\{ \ln \frac{1}{p_i}+\frac{1}{2}(a+b)\left(\ln \frac{1}{p_i}\right)^2+\frac{1}{6}\left(a^2+ab+b^2\right)  \left(\ln \frac{1}{p_i}\right)^3 \ldots \right\}.
\eeq

g) The original notion of \textit{group entropies} was introduced in \cite{Tempesta4} to describe an infinite family of weakly composable entropies. They are defined by
\beq
S_{G}(p):=k_{\mathcal{B}} \sum_{i=1}^{W}p_i Log_{G}\left(\frac{1}{p_i}\right). \label{GE}
\eeq
Here $Log_{G}$ denotes the generalized logarithm
\begin{eqnarray}
\noindent  Log_{G}(x)=\frac{1}{\sigma}\sum_{n=l}^{m}k_{n}x^{\sigma n}, \quad l, m\in\mathbb{Z},
 \quad m-l=r> 0, \quad x>0 \label{GenLog}
\end{eqnarray}
where $k_{n}$ are real constants such that
\begin{equation}
\sum_{n=l}^{m}k_{n}=0, \quad\sum_{n=l}^{m} n k_{n}=1, \label{2.10}
\end{equation}
and $k_m\neq0$, $k_l\neq0$. Conditions \eqref{2.10} are sufficient to ensure that  $\lim_{\sigma\rightarrow 0} Log_{G}(x)=\ln x$.  The class \eqref{GE} is a subclass of \eqref{entropy}, corresponding to the choice
\beq
G(t)=\frac{1}{\sigma}\sum_{n=l}^{m}k_{n} \exp(n \sigma t)\text{,}\quad l\text{,
}m\in\mathbb{Z}\text{,}\quad m-l=r> 0\text{,}\label{2.9}
\eeq
with the constraints \eqref{2.10}. 

\noi Notice that there is much freedom in choosing the coefficients $k_n$, since only two conditions are required to guarantee the correct limit.

Due to the specific choice of the formal group structure, the generalized logarithm \eqref{GenLog} is intimately related to a family of discrete derivatives. Indeed, let us denote by $T$ the shift operator, acting on a function $f$ as $T f(x)= f(x+\sigma)$. The discrete derivatives of order $r$ associated to the generalized logarithms \eqref{GenLog} are defined to be
\[
\Delta_r(\sigma)= \frac{1}{\sigma}\sum_{n=l}^{m}k_n T^{n}, \qquad \quad r=m-l.
\]
In the limit $\sigma\rightarrow 0$, $\Delta_r(\sigma)= \partial_x$. For a suitable choice of the values of the parameter $\sigma$, each of the obtained entropies is concave, and weakly composable.

\noi The relation between group entropies and the class $S_{c,d}$ was an open question. Actually, the two classes are simply \textit{different}. In other words, the entropies \eqref{GE}, with the exception of Tsallis entropy, cannot be obtained from $S_{c,d}$ entropy by specializing the coefficients $(c,d)$.

\noi The first representative of the group-entropy class is indeed the Tsallis entropy, obtained for $Log_{G}(x) = \frac{x^{\sigma}-1}{\sigma}= Log_{\mathcal{T}}(x)=\frac{x^{1-q}-1}{1-q}$ (here $\sigma=1-q$). The second representative of the class, the Kaniadakis entropy, is not in the $S_{c,d}$ family.

\noi Consider the logarithm $Log_{III}(x)=\frac{ x^{1-q}-2 x^{-(1-q)}+x^{-2(1-q)}}{1-q}$ (associated to a third order discrete derivative). The corresponding entropy is
\beq
S_{III}:=\frac{k_{\mathcal{B}}}{1-q} \sum_{i=1}^{W} p_i \left( p_{i}^{2 (1-q)} -2 p_i^{(1-q)} +p_{i}^{-(1-q)} \right) \label{SIII}
\eeq
This entropy  is concave for $2/3<q<1$ (here we prefer to use the parameter $q=1-\sigma$). In the limit $q\rightarrow 1$, the entropy $S_{III}$ reduces to the Boltzmann-Gibbs entropy. Also, the entropy $S_{III}$ is weakly composable, since it corresponds to $G(t)= \frac{e^{(1-q)t}-2 e^{-(1-q)t}+e^{-2(1-q)t}}{1-q}$.  However, it is not obtainable neither from Borges-Roditi's nor from $S_{c,d}$ entropy. Its expansion is given by
\[
S_{III}[p]=k_{\mathcal{B}} \sum_{i=1}^{W} p_i \left\{ \ln \frac{1}{p_i}+\frac{3}{2}(1-q)\left(\ln \frac{1}{p_i}\right)^2-\frac{5}{6}\left(1-q\right)^2  \left(\ln \frac{1}{p_i}\right)^3 \ldots \right\}
\]
\begin{remark}
Entropy \eqref{SIII} provides an interesting example of a functional which is in the universal class \eqref{entropy}, but it does not satisfy condition \eqref{conc}. Indeed, this condition, although sufficient, is not necessary to define an entropic functional with good thermodynamic properties.
\end{remark}
\noi Another entropy, with very similar properties, is the following one, corresponding to a fourth order discrete derivative:
\[
S_{IV}:=\frac{k_{\mathcal{B}}}{1-q} \sum_{i=1}^{W} p_i \left( p_{i}^{-2 (1-q)} -\frac{3}{2} p_i^{-(1-q)} +\frac{3}{2} p_i^{(1-q)}-p_{i}^{2(1-q)} \right).
\]
\noi By using the construction proposed in \cite{Tempesta4}, infinitely many new entropies, each of them depending on a parameter $q$, and defined in a specific interval of values of $q$, can be constructed. They are not particular cases of any of the previously discussed entropies. At the same time, being trace-form and satisfying the axioms (SK1)--(SK3), they comply with the two scaling laws postulated in \cite{Hanel} (as is also easy to prove directly).

\noi Both families, i.e. group entropies and the $S_{c,d}$ entropies can be interpreted in the universal group-theoretical framework. Indeed,  the  $S_{c,d}$ entropy for specific choices of $(c,d)$ is expressible in terms of the  entropy \eqref{entropy} by means of the general formula \eqref{genexp} and is consequently weakly composable. The group entropies \eqref{GE} by construction are a subfamily of the entropy \eqref{entropy}.

\section{A new three-parametric group entropy: the $S_{q,\alpha,\beta}$ entropy}
To illustrate the potential richness of the theory previously developed, we wish to present here a new nontrivial multi-parametric entropic functional, obtained as a special case of the construction sketched above. It reads
\bea
\nn S_{\alpha, \beta, q}[p]:=\frac{k_{\mathcal{B}}}{1-q} \sum_{i=1}^{W} p_i \bigg( \alpha p_{i}^{-2 (1-q)} +\frac{1}{2}\left(1-3\alpha +\beta\right) p_i^{-(1-q)} +\\  \frac{1}{2}\left(\alpha-1-3\beta\right) p_i^{(1-q)}+\beta p_{i}^{2(1-q)} \bigg). \label{Sqab}
\eea
The main properties of this entropy are the following.

\noi $i)$
\[
\lim_{q\rightarrow 1} S_{\alpha, \beta, q}[p]= S_{BG}[p]
\]
irrespectively of the choice of the parameters.

\noi $ii)$ The entropy \eqref{Sqab} is trace-form.

\noi $iii)$ It satisfies the first three SK axioms; in particular, it is concave, for instance, for $1/2<q<3/2$, $0<\alpha<1/4$, $-1/4<\beta<0$.

\noi $iv)$ The parameters $\alpha, \beta, q$ are independent; consequently, the entropy \eqref{Sqab} is not reducible to the previously discussed two-parametric classes.

\noi $v)$ It is weakly composable, with expansion
\bea
\nn S_{\alpha, \beta, q}[p]:=k_{\mathcal{B}} \sum_{i=1}^{W} p_i \bigg\{ \ln \frac{1}{p_i}+\frac{3}{2}(\alpha+\beta)(1-q)\left(\ln \frac{1}{p_i}\right)^2+ \\  \frac{1}{6}\left(1+6\alpha-6\beta\right)\left(1-q\right)^2  \left(\ln \frac{1}{p_i}\right)^3 +\ldots \bigg\} \label{SIV}
\eea
The entropy \eqref{SIV} is in the group entropy class, but it was not considered explicitly in \cite{Tempesta4}. Also, it is not a particular case of $S_{c,d}$ or $S_{BR}$. This example is just a representative of a large class of entropic functionals, each of them multiparametric, that can be constructed by using the group-theoretical framework previously discussed. However, the complete analysis of this class, which is a specific realization of the notion of universal-group entropy, is outside the scopes of this paper.

\section{Distribution functions and thermodynamic properties}

We shall discuss the maximization of group entropies under appropriate constraints: we adopt a generalized maximum entropy principle (see e.g. \cite{Kaniad2}, \cite{AT}). We will see that the Legendre structure of classical thermodynamics, at least in some aspects, is preserved in our group-theoretical framework.


\noi Precisely, let
\[
Log_{U}[\epsilon]=G\left(\ln \epsilon\right) \label{LogU}
\]
where $G(t)$ is the power series \eqref{I.2}, with the constraint \eqref{conc}. Consider an isolated system in a stationary state (\textit{microcanonical ensemble}). The optimization of $S_U$ leads to the equal probability case, i.e. $p_i=1/W, \forall i$. Therefore, we have
\beq
S_U[p]=k_{\mathcal{B}} Log_U W,
\eeq
which reduces to the celebrated Boltzmann formula $S_{BG}=k_{\mathcal{B}}\ln W$ in the case of uncorrelated particles.

Let us consider a system in thermal contact with a reservoir (\textit{canonical ensemble}). We introduce the numbers $\epsilon_i$, interpreted as the values of a physically relevant observable, typically the value of the energy of the system in its $i$th state.  Assume that $p_i(\epsilon_i)$ is a normalized and monotonically decreasing distribution function of $\epsilon_i$. The \textit{internal energy} $\mathcal{V}$ in a given state is defined as $\mathcal{V}=\sum_{i=1}^{W}\epsilon_i p_i(\epsilon_i)$.

As usual, we shall study the variational problem of the existence of a stationary distribution $\widetilde{p}_i(\epsilon)$. However, this analysis can not be performed in direct analogy with the standard case. We introduce the functional
\beq
L=S_{G}[p]-\alpha\left[\sum_{i}p(\epsilon_i)-1\right]-\beta\left[\sum_{i=1}^{W}\epsilon_i p_i(\epsilon_i)-\mathcal{V}\right], \label{Leg}
\eeq
where $\alpha$ and $\beta$ are Lagrange multipliers. The vanishing of the variational derivative of this functional with respect to the distribution $p_i$ provides the stationary solution
\beq
\widetilde{p}_i=\frac{E(-\alpha-\beta \epsilon_i)}{Z},
\eeq
with $Z=\sum_{i=1}^{W}E(-\alpha-\beta \epsilon_i)$, and $E(\cdot)$ is an invertible function.


\noi However, as already pointed out in \cite{Tempesta4}, only in particular cases we are able to identify $E$ with the inverse of a generalized logarithm $Log_{G}$.



In \cite{Kaniad2}, \cite{KLS}, the class of entropies allowing a treatment with the variational approach described above has been determined. In the Legendre-Massieu framework, one can derive the interesting relation
\beq
Log_{G}(Z)+\beta \mathcal{V}=S_{G}.
\eeq
The previous equation can be used to introduce a thermodynamic observable $T$, which has the interpretation of a local temperature for a non-equilibrium metastable state. Precisely, we can define it from $
\frac{\partial S_U}{\partial \mathcal{V}}=\frac{1}{T}.
$
Analogously, a generalized free energy can be introduced according to $\mathcal{F}=\mathcal{V}-TS_{U}$. However, in the context of Tsallis entropy, the variational approach leading to the definition of an effective temperature is based on more specific constraints (escort distributions \cite{Tbook}).

The determination of a group-theoretical procedure to construct the appropriate constraints for an entropy of the class \eqref{entropy}   is an open problem.
\section{On the asymptotic behaviour of generalized entropies}

From a mathematical point of view, the study of the asymptotic behaviour of a given entropy, in the limit of large size systems, is not a well defined task, even for the standard Boltzmann-Gibbs entropy.

Let $W$ be the number of microscopic states admitted by a system.  We shall focus on the case of large size systems (under the hypothesis that $W\rightarrow\infty$ for $N\to \infty$). An entropy is then a functional $S=S[p]$ defined on the space $\mathcal{P}_{\infty}$.

A simple argument proves that, depending on the choice of the distribution in $\mathcal{P}_{\infty}$, we can get infinitely many different limits even for the $S_{BG}$ case. Indeed, consider the probability distribution $\{p\}=(1,0,0,\ldots)$ with infinitely many entries. Then $S_{BG}[p]=0$. If $\{p\}=(1/2,1/2,0,0,\ldots)$, then $S_{BG}[p]=\ln 2$ (in units of $k_{\mathcal{B}}$). Let $\mathcal{N}=10^{80}$ (i.e., the estimated number of atoms in the observable Universe), and $\{p\}= (\underbrace{1/\mathcal{N}, 1/\mathcal{N}, \ldots, 1/\mathcal{N}}_{\mathcal{N}-times}, 0,0, \ldots)$, then $S_{BG}[p]=80 \cdot \ln 10\simeq 184.2$. From the point of view of probability and information theory, there is no way to have uniqueness of the limit on the full space $\mathcal{P}_{\infty}$.

In the domain of classical thermodynamics, a priori the same objection applies. However, if we accept to restrict  to the important case of the uniform distribution (which is not the only physically interesting case), then one can define properly a \textit{thermodynamic limit of an entropy on the uniform distribution}. From a physical perspective, this would correspond to the case of a double limit, both of large system size and of large times.


The only meaningful way to compare the behaviour of different entropies is to compute them all  \textit{in the same regime}. Once we accept to restrict to uniform probabilities, then the known entropies assume the form of the following asymptotic functions: either $(ln W)^{a}$ or $W^{b}$ (possibly multiplied by parameters), or the product of these forms. Presently, we cannot exclude that other forms could also be possible.

All this is a consequence of the analysis of scaling laws performed in \cite{Hanel}. However, it must be noticed that these simple functions of $W$ are not entropic functionals, and they cannot be obtained from a known entropy by specializing its parameters.



The problem of ``comparing" entropies in the large size limit finds its most simple answer in the regime, if exists, where they are extensive. In classical thermodynamics, the only reason to consider generalized entropies is the fact that they can be \textit{extensive} in regimes where the $S_{BG}$ entropy is not. This is the truly important limiting property.


What is crucial is that \textit{all admissible entropies have the same behavior} in the  regime where they are extensive, i.e. proportional to the number $N$ of particles of the system.

Therefore, it is not  surprising that different entropies could share the same asymptotic behavior.



When we are interested in more general contexts as probability theory and information theory, then all of the infinitely many possible distributions are a priori relevant. Therefore, according to the previous discussion, the comparison among asymptotic behaviors of different entropies loses its meaning.

\section{Open problems and future perspectives}

The main message of this paper is that \textit{the notion of entropy is intimately related with group theory}. Indeed, admissible entropies have associated a group law, controlling crucially their properties.

The universal-group entropy is a flexible tool, providing new insight in different contexts of the theory of complex systems. For instance, it could offer a more general approach to the description of the evolution of biological complexity \cite{AOC}.
The role of the universal-group entropy in geometric information theory should be properly clarified. New examples of Kullback-Leibler divergences are presently under investigation. They could provide refined tests of similarity, for instance, in the analysis of genomic sequences. A nonadditive entropy taken from the class \eqref{entropy} a priori could be more appropriate than the Shannon one, in order to take into account the so-called \textit{epistatic correlations} between sites in a genomic chain.

An open problem is to find constructively classes of new entropies obtained as special cases of the general definition \eqref{entropy}. New results along these lines are contained in \cite{Tempestaprepr}.

Another important aspect that deserves to be clarified is the possible number-theoretical content of the notion of universal group entropy. In \cite{Tempesta4}, a connection between a set of entropic functionals and a class of Dirichlet series has been established. It has been proved that the formal group exponentials defining a specific family of entropies can be used to define $L$-functions. The most paradigmatic case is offered by Tsallis entropy, directly related to the Riemann zeta function. A construction of universal congruences and generalized Bernoulli polynomials connected with formal groups has been provided in \cite{Tempesta5}
and in \cite{Tempesta2}. It would be interesting to realize explicitly this picture in the general group-theoretical framework developed in this work.

\appendix
\section{The Shannon-Khinchin axioms}
In the original formulation  due to Khinchin \cite{Khinchin}, the Shannon-Khinchin axioms were called ``properties"  that guarantee Khinchin's uniqueness theorem for the Boltzmann-Gibbs entropy $S$. In particular, the properties of continuity and maximum were unified into a unique statement. Here we propose a modern version, in terms of four requirements.

(SK1) (Continuity). The function $S(p_1,\ldots,p_W)$ is continuous with respect to all its arguments

(SK2) (Maximum principle).  The function $S(p_1,\ldots,p_W)$ takes its maximum value for the uniform distribution $p_i=1/W$, $i=1,\ldots,W$.

(SK3) (Expansibility). Adding an impossible event to a probability distribution does not change its entropy: $S(p_1,\ldots,p_W,0)=S(p_1,\ldots,p_W)$.

(SK4) (Additivity). Given two subsystems $A$, $B$ of a statistical system,
\[
S(A \cup B)=S(A)+S(B\mid A);
\]
here $S(B\mid A)$ denotes the conditional entropy associated with the conditional distribution $p_{ij}(B\mid A)$.

\vspace{3mm}

\textbf{Acknowledgments}.
I wish to thank heartily prof. C. Tsallis for a careful reading of the manuscript and many useful discussions, and prof. G. Parisi for reading the manuscript and for encouragement. Interesting discussions with prof. A. Gonz\'alez-L\'opez, R. A. Leo and M. A. Rodr\'iguez are also gratefully acknowledged.

This work has been supported by the research project FIS2011--22566, Ministerio de Ciencia e Innovaci\'{o}n, Spain.

\end{document}